\documentclass[journal,twoside,web]{ieeecolor}
\usepackage{lcsys}

\usepackage{comment}
\usepackage{cite}
\usepackage{booktabs}

\usepackage{multicol}

\usepackage{mathtools}
    
\usepackage{amsmath}
\usepackage{amstext}
\usepackage{amssymb}
\usepackage{amsfonts}
\usepackage{float}

\usepackage{amsthm}  
\newtheorem{theorem}{Theorem}
\newtheorem{lemma}{Lemma}
\newtheorem{definition}{Definition}

\newtheorem{assumption}{Assumption}
\newtheorem{problem}{Problem}
\makeatletter

\newcommand{\rmnum}[1]{\romannumeral #1}
\newcommand{\Rmnum}[1]{\expandafter\@slowromancap\romannumeral #1@}
\usepackage{savesym}

\usepackage{algorithm}
\usepackage{algorithmicx} 
\usepackage{algpseudocode} %
\savesymbol{AND}
\usepackage[group-separator={,},group-minimum-digits={3}]{siunitx}

\usepackage{graphicx} 
\usepackage{epsfig} 

\usepackage{times} 
\usepackage{amsmath} 
\usepackage{amssymb}  
\usepackage{comment}
\makeatletter
\let\NAT@parse\undefined
\makeatother
\usepackage{hyperref}
\hypersetup{
   colorlinks=true,
    linkcolor= blue,
    allcolors=blue,
    citecolor = blue,
    filecolor=black,      
    urlcolor=blue,
    }
\usepackage{mathrsfs}

\begin{document}
\title{
A Priority-Aware Replanning and Resequencing Framework for Coordination of Connected and Automated Vehicles}
\author{Behdad Chalaki, \IEEEmembership{Student Member, IEEE}, Andreas A. Malikopoulos, \IEEEmembership{Senior Member, IEEE}%

\thanks{This research was supported by ARPAE's NEXTCAR program under the award number DE-AR0000796. This support is gratefully acknowledged.}
\thanks{The authors are with the Department of Mechanical Engineering, University of Delaware, Newark, DE 19716 USA (emails: \texttt{\{bchalaki;andreas\}@udel.edu}).}}
\maketitle
\thispagestyle{empty}

\begin{abstract} 
Deriving optimal control strategies for coordination of connected and automated vehicles (CAVs) often requires re-evaluating the strategies in order to respond to unexpected changes in the presence of disturbances and uncertainties. In this paper, we first extend a decentralized framework that we developed earlier for  coordination of CAVs at a signal-free intersection to incorporate replanning. Then, we further enhance the framework by introducing a priority-aware resequencing mechanism which designates the order of decision making of CAVs based on theory from the job-shop scheduling problem. Our enhanced framework relaxes the first-come-first-serve decision order which has been used extensively in these problems. We illustrate the effectiveness of our proposed approach through numerical simulations.
\end{abstract}

\begin{IEEEkeywords}
connected and automated vehicles, replanning, resequencing, sequential decision making
\end{IEEEkeywords}

\section{Introduction}
\IEEEPARstart{S}{everal} research efforts in the literature have considered a two-level optimization framework for coordination of connected and automated vehicles (CAVs) at traffic bottlenecks.
An \emph{upper-level} optimization yields, for each CAV, the optimal time to exit the control zone, while a \emph{low-level} optimization yields for the CAV its optimal control input (acceleration/deceleration) to achieve the optimal time derived in the upper-level subject to the state, control, and safety constraints. 
There have been several approaches in the literature to solve the upper-level optimization problem, including first-in-first-out (FIFO) queuing policy 
\cite{Malikopoulos2017,Rios-Torres2017}, heuristic  Monte Carlo tree search methods \cite{xu2019cooperative,xu2020bi}, centralized optimization techniques \cite{guney2020scheduling,hult2018optimal}, and job-shop scheduling \cite{chalaki2020TITS, fayazi2018mixed}.
Given the solution of the upper-level optimization problem, the constrained optimal control problem is solved sequentially in the low-level optimization, yielding the optimal control input for each CAV. 
To solve the low-level optimization problem, research efforts have  used  optimal control techniques to derive the closed form solutions \cite{Malikopoulos2017,chalaki2020TCST,Malikopoulos2020,zhang2019decentralized}, or model predictive control (MPC) \cite{hult2018optimal,kim2014mpc,campos2014cooperative,kloock2019distributed}.  

To the best of our knowledge, there have been limited studies in exploring the effects of decision-making sequence in the low-level optimization problem. 
Campos et al. \cite{de2013autonomous} presented a heuristic approach to find a decision order for CAVs at an intersection based on their time to reach an unsafe set, i.e., the CAV can no longer stop before the intersection. Alrifaee et al. \cite{alrifaee2016coordinated} proposed a graph-based approach to construct levels of parallelizable agents for non-cooperative decentralized MPC, in which agents on the same level solve the problem in parallel, sequentially after agents on the previous level. Xiao and Cassandras \cite{xiao2020decentralized}
 relaxed the FIFO queuing policy by formulating a resequencing problem before the control zone. The authors assumed that after a CAV performs the resequencing, its speed remains constant until it arrives at the control zone.  

In this paper, we build upon the framework introduced in \cite{Malikopoulos2020} consisting of a single optimization level aimed at both minimizing energy consumption and improving the traffic throughput. Using the proposed framework, each CAV computes the optimal exit time corresponding to an unconstrained energy optimal trajectory which satisfies all the state, control, and safety constraints. We extend this work by integrating the replanning mechanism into the framework. Since unexpected changes in the presence of disturbances and uncertainties can result in deviations from the optimal planned trajectory of the CAVs, the replanning mechanism introduces feedback in the planning which can respond to these changes in the system to some extent. In addition, using the theory from the job-shop scheduling problem, we further enhance the framework by introducing a priority-aware resequencing mechanism to find the decision making sequence of the CAVs based on the minimum exit time from the traffic bottleneck. 

The work that we report on this paper advances the state of the art in a way that relaxes the first-come-first-serve (FCFS) decision making sequence of the CAVs. The contributions of this paper are: (\rmnum{1})  the introduction of replanning as a feedback mechanism to handle uncertainties or disturbances, and (\rmnum{2}) the development of a priority-aware resequencing mechanism for the coordination of CAVs.

The remainder of the paper is structured as follows. 
In Section \ref{sec:pf}, we introduce the modeling framework and we present the priority-aware resequencing mechanism in Section \ref{sec:RS}. Finally, we provide simulation results in Section \ref{sec:Results}, and concluding remarks in Section \ref{sec:Conclusion}.

\section{Modeling Framework} \label{sec:pf}
We consider a signal-free intersection (Fig. \ref{fig:intersection}), which includes a \textit{coordinator} that stores information about the intersection's geometry and CAVs' trajectories. The coordinator acts as a database for the CAVs and does not make any decision. The intersection includes a \textit{{control zone}} inside of which the CAVs can communicate with the coordinator. We call the points inside the control zone where paths of CAVs intersect and a lateral collision may occur as \textit{conflict points}. Let $\mathcal{O}\subset \mathbb{N}$ index the set of conflict points, $N(t)\in\mathbb{N}$ be the total number of CAVs inside the control zone at time $t\in\mathbb{R}_{\geq0}$, and $\mathcal{N}(t)=\{1,\ldots,N(t)\}$ be the queue that designates the order in which each CAV entered the control zone. We model the dynamics of each CAV $i\in\mathcal{N}(t)$ as a double integrator
\begin{align}
\begin{aligned}\label{eq:dynamics}
\dot{p}_i(t)=v_i(t),\\
\dot{v}_i(t)=u_i(t),
\end{aligned}\end{align}
where $p_{i}(t)\in\mathcal{P}_{i}$, $v_{i}(t)\in\mathcal{V}_{i}$, and
$u_{i}(t)\in\mathcal{U}_{i}$ denote position, speed, and control input at $t$, respectively. The sets $\mathcal{P}_{i}$,
$\mathcal{V}_{i}$, and $\mathcal{U}_{i}$, for $i\in\mathcal{N}(t),$
are compact subsets of $\mathbb{R}$. Let $t_{i}^{0}\in\mathbb{R}_{\geq 0}$ be the time that CAV $i\in\mathcal{N}(t)$
enters the control zone, and $t_{i}^{f}>t_i^0\in\mathbb{R}_{\geq 0}$ be the time that CAV $i$ exits the control zone. For each CAV $i\in\mathcal{N}(t)$, the control input and speed are bounded by 
\begin{align}
    u_{i,\min}&\leq u_i(t)\leq u_{i,\max}, \label{eq:uconstraint} \\
    0< v_{\min}&\leq v_i(t)\leq v_{\max} \label{eq:vconstraint},
\end{align}
where $u_{i,\min},u_{i,\max}$ are the minimum and maximum control inputs and $v_{\min},v_{\max}$ are the minimum and maximum speed limits, respectively. 

\begin{figure}
    \centering
\includegraphics[width=0.8\linewidth]{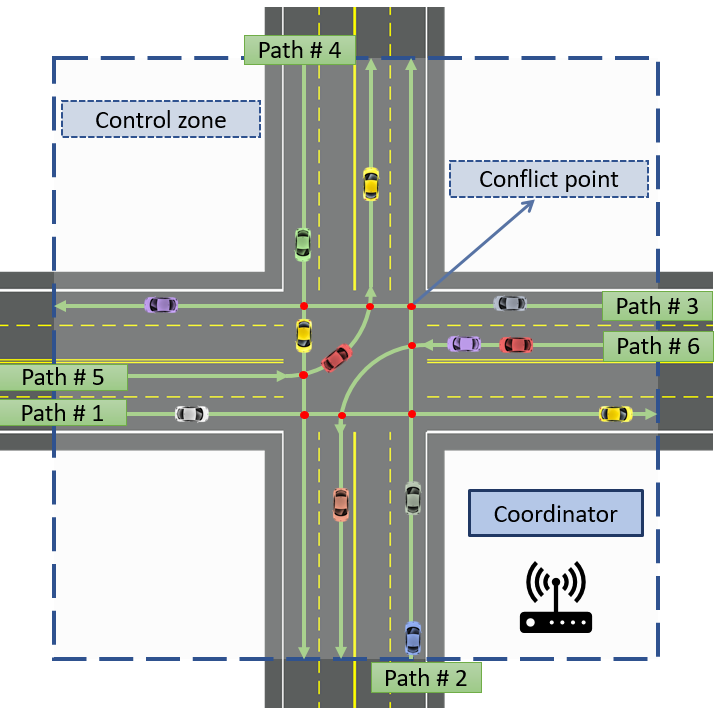}  \caption{A signal free intersection with conflict points.}
    \label{fig:intersection}
\end{figure}

We consider that CAVs do not perform any lane-change maneuver, and thus there are finite paths among which they can choose. The set of all possible paths in the control zone is given by $\mathcal{L}=\{1,\dots,z\}$, $z\in\mathbb{N}$. The path of the CAV $i\in\mathcal{N}(t)$ in the control zone is denoted by $\ell_i\in\mathcal{L}$ (Fig. \ref{fig:intersection}), and is decided a priori based on some upper-level routing problem. In our modeling framework, we make the following assumption.
\begin{assumption}\label{assumption:Paths}
Each path of a CAV cannot get either split to two paths or merged by another CAV's path. 
\end{assumption}
This assumption implies that there should be separate lanes for the turning maneuvers at the intersections. This might be a strong assumption, but it simplifies the complexity of our resequencing algorithm (formally defined next) by only considering the rear-end safety constraints for the CAVs travelling on the same path. Investigating the implications of relaxing this assumption is the subject of ongoing research.

To guarantee rear-end safety between CAV $i\in\mathcal{N}(t)$ and a preceding CAV $k\in\mathcal{N}(t)\setminus\{i\}$, we impose the following constraint,
\begin{gather}\label{eq:rearend}
 p_k(t)-p_i(t)\geq \delta_i(t) =\gamma + \varphi\cdot v_i(t),
\end{gather}
where $\delta_i(t)$ is the safe speed-dependent distance, while $\gamma$ and $\varphi\in\mathbb{R}_{>0}$ are the standstill distance and reaction time, respectively. 

Let CAV $k\in\mathcal{N}(t) \setminus \{i\}$ be a CAV that has already planned its trajectory which might cause a lateral collision with CAV $i$. 
We denote by $p_{i}^n$ and $p_{k}^n$ the distance of the conflict point $n\in\mathcal{O}$ from $i$'s and $k$'s paths' entries, respectively.
Since we do not use the FIFO queuing policy, CAV $i$ can reach at conflict point $n$ either after or before CAV $k$. 
In the first case, i.e., when CAV $i$ reaches at conflict point $n$ after CAV $k$, we have
\begin{equation} \label{eq:lateralAfter}
    p_i^n - p_i(t) \geq \delta_i(t), \quad \forall t\in[t_i^0, t_k^n],
\end{equation}
where $t_k^n$ is the known time that CAV $k$ reaches at conflict point $n$, i.e., position $p_k^n$.
In the second case, where CAV $i$ reaches at the conflict point $n$ before CAV $k$, we have 
\begin{equation} \label{eq:lateralBefore}
    p_k^n - p_k(t) \geq \delta_k(t)=\gamma + \varphi\cdot v_k(t), \quad \forall t\in[t_k^0, t_i^n],
\end{equation}
where $t_i^n$ is determined by the trajectory planned by CAV $i$.  
Since $ 0 < v_{\min} \leq v_i(t)$, the position $p_i(t)$ is a strictly increasing function.
Thus, the inverse $t_i\left(\cdot\right) = p_i^{-1}\left(\cdot\right)$ exists and it is called the \textit{time trajectory} of CAV $i$ \cite{Malikopoulos2020}. Hence, we have $t_i^n = p_i^{-1}\left(p_i^n\right)$. Therefore, for each candidate path of CAV $i$, there exists a unique time trajectory which can be evaluated at conflict point $n$, i.e.,  $t_i(p_i^n)$, to find the time that CAV $i$ reaches at conflict point $n$, i.e., $t_i^n$.

By moving all terms in \eqref{eq:lateralAfter} to the RHS, we get $\delta_i(t) + p_i(t) - p_i^n\leq 0$. Constraint \eqref{eq:lateralAfter} is satisfied, if $\max(\delta_i(t) + p_i(t) - p_i^n)\leq 0$ in the interval $[t_i^0, t_k^n]$. Likewise, if $\max(\delta_k(t) + p_k(t) - p_k^n)\leq 0$ in the interval $[t_k^0, t_i^n]$ constraint \eqref{eq:lateralBefore} is satisfied.
However, to ensure the lateral safety between CAV $i$ and CAV $k$ at conflict point $n$, either \eqref{eq:lateralAfter} or \eqref{eq:lateralBefore} must be satisfied, and thus we impose the lateral safety constraint on CAV $i$ using minimum function as
\begin{align}
    \min \Bigg\{ &\max_{t\in[t_i^0, t_k^n]} \{ \delta_i(t) + p_i(t) - p_i^n\}, \notag\\
            &\max_{t\in[t_k^0, t_i^n]} \{ \delta_k(t) + p_k(t) - p_k^n \}   \Bigg\} \leq 0. \label{eq:lateralMinSafety}
\end{align}

In our framework, each CAV $i$ communicates with the coordinator to solve a time minimization problem, which determines $t_i^f$, i.e., the time that CAV $i$ must exit the control zone. The time $t_i^f$ corresponds to the unconstrained energy optimal trajectory guaranteeing that state, control, and safety constraints are satisfied. This trajectory is communicated back to the coordinator, so that the subsequent CAVs receive this information and plan their trajectories accordingly. 
Our framework implies that the CAVs do not have to come to a full stop at the intersection, thereby conserving momentum and fuel while also improving travel time. By enforcing the unconstrained energy-optimal trajectory that guarantees the satisfaction of all the state, control, and safety constraints, we avoid inherent real-time implementation difficulties in solving a constrained optimal control and piecing constrained and unconstrained arcs together \cite{Malikopoulos2020,chalaki2020experimental}. 

We start our exposition with the unconstrained energy optimal solution of CAV $i$, which has the following form \cite{Malikopoulos2017}
\begin{align} \label{eq:optimalTrajectory}
    u_i(t) &= 6 a_i t + 2 b_i, \notag \\
    v_i(t) &= 3 a_i t^2 + 2 b_i t + c_i, \\
    p_i(t) &= a_i t^3 + b_i t^2 + c_i t + d_i, \notag
\end{align}
where $a_i, b_i, c_i, d_i$ are constants of integration. CAV $i$ must also satisfy the boundary conditions
\begin{align}
     p_i(t_i^0) &= 0,\quad  v_i(t_i^0)= v_i^0 , \label{eq:bci}\\
     p_i(t_i^f)&=p_i^f,\quad u_i(t_i^f)=0, \label{eq:bcf}
\end{align}
where $u_i(t_i^f)=0$ because the speed at the exit of the control zone is not specified \cite{bryson1975applied}.
The details of the derivation of the unconstrained solution are discussed in \cite{Malikopoulos2017}.  

One of the advantages of incorporating replanning in the framework is introducing feedback in the system.
Replanning can occur either periodically (i.e., at a period determined a priori) or be event-driven (i.e., based on an occurrence of a certain event such as the entrance of a new CAV in the control zone). All CAVs in the control zone observe their state at each replanning instance and re-solve their optimization problem, discussed next, sequentially with the new initial conditions. For CAV $i$, let $\tau\in[t_i^0,t_i^f]$ be the replanning time, and $\tilde{\mathbf{x}}_{i}(\tau)=\left[\Tilde{p}_{i}(\tau)~ \tilde{v}_{i}(\tau)\right]^\top$, be the measurement of the state at this time. The revised initial conditions for CAV $i$ at this replanning instance is given by

\begin{align}
     p_i(\tau) &= \Tilde{p}_{i}(\tau),\quad  v_i(\tau)= \Tilde{v}_{i}(\tau).  \label{eq:bciNEW}
\end{align}

\begin{definition}
The compact set $\mathcal{T}_i(\tau)=[\underline{t}_i^{f,{\tau}}, \overline{t}_i^{f,{\tau}}]$ is the set of feasible solution of CAV $i\in\mathcal{N}(t)$ for the exit time, where $\underline{t}_i^{f,{\tau}}$ and $ \overline{t}_i^{f,{\tau}}$ denotes the minimum and maximum feasible exit time computed at $\tau$. CAV $i$ can determine $\mathcal{T}_i(\tau)$ at time $\tau$ using the speed and control input constraints \eqref{eq:uconstraint}-\eqref{eq:vconstraint}, initial condition \eqref{eq:bci} or \eqref{eq:bciNEW} (depending on if $\tau=t_i^0$), and final condition \eqref{eq:bcf}. For the derivation of this compact set, refer to \cite{chalaki2020experimental}.
\end{definition}

To avoid abrupt changes in the control input and unnecessary acceleration, we revise the lower bound on exit time to use the maximum value between the earliest feasible exit time computed at $t_i^0$ (to simplify the notation denoted as $\underline{t}_i^{f}$), and the earliest feasible exit time computed at $\tau$. Thus, the feasible compact set computed at $\tau$ is given by

\begin{equation}
    \mathcal{T}_i(\tau)=\left[\max\left\{\underline{t}_i^{f},\underline{t}_i^{f,\tau}\right\}, \overline{t}_i^{f,{\tau}}\right].
\end{equation}

\begin{problem}\label{pb:timeMinRP}
Each CAV $i\in\mathcal{N}(t)$ at replanning instance $\tau$ solves the following optimization problem
\end{problem}

\begin{align}\label{eq:tif}
    &\min_{t_i^f\in \mathcal{T}_i(\tau)} t_i^f \\
    \text{subject to: }& \notag\\
    & \eqref{eq:rearend}, \eqref{eq:lateralMinSafety}, \eqref{eq:optimalTrajectory} \notag.
\end{align}

To some extent, this replanning provides  CAV a feedback mechanism to react to any uncertainties. Ongoing research analyzes the uncertainties and consider these in the planning of CAVs \cite{chalaki2021RobustGP}. 

\section{A priority-aware Resequencing}\label{sec:RS}
In our previous framework \cite{Malikopoulos2020}, upon entering the control zone, CAV $i\in\mathcal{N}(t)$ solves Problem \ref{pb:timeMinRP} at $\tau=t_i^0$ by only considering CAVs in the control zone. For the cases in which CAVs enter the control zone simultaneously, the coordinator randomly decides the decision-making order of CAVs. Namely, the order of decision making is based on the order that CAVs entered the control zone, FCFS. 
We define the \textit{decision sequence} formally as follows.

\begin{definition}
The sequential decision making of $N(t)$ CAVs is based on the decision sequence that is given by the sequence $s =(s_1,s_2,\dots,s_{N(t)})$ where $s_{n} \in\mathcal{N}(t)$, ${n}\in\{1,\dots,N(t)\}$ is the ${n}$'th CAV in the decision making process.
\end{definition}

Without a resequencing mechanism, the decision sequence of the $N(t)$ CAVs is given by $s=(1,2,3,\dots,N(t))$ which is imposed by the order the CAV enter the control zone, referred to as FCFS sequence. Note that this is different from the order that CAVs cross the intersection, which is determined by the lateral safety constraint \eqref{eq:lateralMinSafety}. Next, we introduce our resequencing framework, which designates the decision sequence at each replanning instance.

Unlike our previous framework, where CAVs only solve their optimization problem upon entering the control zone, in this enhanced framework, CAVs re-solve the optimal control problem at different instances based on new observed information. {The observed information of each CAV consists of position and speed of the CAV at the replanning instance, which then can be used as new initial conditions \eqref{eq:bciNEW} to solve Problem \ref{pb:timeMinRP}.} In this section, we introduce a priority-aware resequencing mechanism to find the sequence of decision making based on the minimum exit time from the control zone using scheduling theory.

A scheduling problem is shown by a triplet ($\alpha~|~ \beta~ |~ \gamma$), where $\alpha$ and $\beta$ fields describe the machine environment and details of the processing characteristics and constraints, respectively, while $\gamma$ field describes the objective function. In our problem, the control zone can be considered as \textit{a single machine}, while different CAVs are considered as different jobs. In our problem, we also have \textit{precedence constraint} which requires that a CAV not plan earlier than the physical CAV located in front of it, which we define formally next. 
 \begin{definition}
 The precedence constraint can be represented by a directed graph $G = (V, E)$, where $V\coloneqq \mathcal{N}(t)$ is set of all CAVs and $E\coloneqq\{(i,j)|i,j \in V, i\rightarrow j\}$ is the set of all constraints on the order of decision making. Edge $(i,j) \in {E}$  shows that CAV $i$ should plan earlier than CAV $j$.   
 \end{definition}
 
 \begin{definition}
A non-empty subgraph $G_1=(V_1,E_1)$, where $V_1\subset V$ and $E_1\subset E$ is called a \textbf{chain} if for each vertex $i \in V_1$, there exist at most a single edge $(i,j)\in E_1,j\in V_1\setminus\{i\}$.
\end{definition}

In a scheduling problem, the \textit{processing time} of a single machine on the job $i$ is denoted by $P_i$, representing the time that it takes for the machine to process job $i$. In our case, we consider the processing time of CAV $i$ at replanning instance $t$ to be equal to $\min(\mathcal{T}_i(t))$, i.e., the minimum exit time from the control zone which is independent of the decision sequence. For each job $i$, a weight $w_i\in\mathbb{R}_{>0}$ describes the importance of job $i$ relative to the other jobs in the system. We consider that the weight of each CAV is inversely proportional to the size of the compact set of its feasible solution. This potentially helps CAVs with smaller feasible space to generate their trajectory first.

Since our goal is to find the optimal decision sequence based on the minimum exit time of the CAVs, we consider the \textit{total weighted completion time} of $N(t)$ CAVs denoted by $J^s = \sum\limits_{i=1}^{N(t)} w_i~C_i^s $ as our cost function under decision sequence $s$, where $C_i^s$ is the sum of processing times of CAV $i$ and other preceding CAVs in the decision sequence $s$. For example, suppose for two CAV $i$ and $j$, we have $P_i < P_j$ and $w_i=w_j$. The cost functions for two different decision sequences $s=(i,j)$ and $s^\prime=(j,i)$ are equal to $J^s = w_i\cdot P_i + w_j \cdot (P_i + P_j)$ and $J^{s^\prime} = w_j\cdot P_j + w_i \cdot (P_j + P_i)$, respectively. It is clear that the decision sequence $s$, which prioritize CAV $i$ over CAV $j$, has a lower total cost.

Our scheduling problem is denoted by $(1~|~G~|~\sum\limits_{i=1}^{N(t)} w_i~C_i^s)$ which describes a single machine model with precedence constraint $G$, and the objective is to minimize the total weighted completion time by finding the optimal decision sequence $s$. 

\begin{lemma}
The precedence constraint's graph of CAVs crossing a single intersection given Assumption \ref{assumption:Paths} consists of multiple disjoint chains.
\end{lemma}
\begin{proof}
From Assumption \ref{assumption:Paths}, we do not have any merging or splitting paths, and thus the precedence constraint only exists among CAV $i$ and $j\in\mathcal{N}(t)$ on the same path $\ell\in\mathcal{L}$ such that $\ell_i =\ell_j = \ell $. Thus, among CAVs on path $\ell$ there exist a chain denoted by $G_\ell\subset G$ such that $\bigcup_{x\in\mathcal{L}}G_x = G$.
\end{proof}
\begin{definition}
A $\rho$-factor of the chain $G_{\ell}=\left(V_{\ell},E_{\ell}\right)$, is denoted by $\rho\left(G_\ell\right)\in\mathbb{R}_{>0}$ and for the chain $G_{\ell}$ given by $1\rightarrow 2 \rightarrow\cdots\rightarrow k$ is computed as
\begin{align}
 \rho(G_\ell) &= \max_{a\in\{1,\dots,k\}}     \left(\frac{\sum_{j=1}^{a}~w_j}{\sum_{j=1}^{a}~P_j}\right)=\frac{\sum_{j=1}^{a^{\ast}}~w_j}{\sum_{j=1}^{a^{\ast}}~P_j},
\end{align} 
where $a^{\ast}\in V_{\ell}\subset\mathcal{N}(t)$ is called the CAV that determines the $\rho$-factor of the chain $G_\ell$.
\end{definition}

The interpretation of the CAV $a^{\ast}\in V_{\ell}$ in the above lemma is that the ratio of weight divided by processing time of the CAV in the chain $G_\ell$ is increasing from the first CAV in the chain until CAV $a^{\ast}$. 

\begin{lemma}\label{lem:rofactor}
If CAV $i\in V_{\ell}\subset\mathcal{N}(t)$ determines the $\rho$-factor of chain $G_\ell=(V_{\ell},E_{\ell})$ given by $1 \rightarrow 2 \rightarrow\cdots\rightarrow k$, $\ell\in\mathcal{L}$, then there exist an optimal decision sequence that processes CAVs $1, \dots, i$ one after another without any interruption by CAVs from other chains $G_{\ell^\prime},\ell^\prime \in \mathcal{L}\setminus \{\ell\}$.  
\end{lemma}
\begin{proof}
The proof is by contradiction and is similar to that of \cite[Lemma 3.1.3.]{pinedo2016scheduling} and follows from \cite[Lemma 3.1.2.]{pinedo2016scheduling} by using the results in which it is optimal to process the chain of jobs $1 \rightarrow 2 \rightarrow\cdots\rightarrow k$ before the chain of jobs $k+1 \rightarrow \cdots\rightarrow n$ if $\frac{\sum_{j=1}^{k}~w_j}{\sum_{j=1}^{k}~P_j} > \frac{\sum_{j=k+1}^{n}~w_j}{\sum_{j=k+1}^{n}~P_j}$.
\end{proof}

By our resequencing mechanism, at each instance of replanning, CAV $i$ accesses the coordinator and inquires the decision sequence computed using Algorithm \ref{Alg:Re-sequencing Algorithm}. 

\begin{algorithm}
\small	
 \caption{Re-sequencing Algorithm }
 \hspace*{\algorithmicindent} \textbf{Input:} All available chains $G_\ell=(V_{\ell},E_{\ell})$, $\ell \in \mathcal{L}$\\
 \hspace*{\algorithmicindent} \textbf{Output:} Decision Sequence $s =(s_1,s_2,s_3\dots,s_N)$
 \begin{algorithmic}[1]
 \While{$\bigcup_{l\in\mathcal{L}}G_l$ is not empty}
\State {Update $a_\ell$ and $\rho \left(G_\ell\right)$ for all $G_\ell=(V_{\ell},E_{\ell})$, $\ell \in \mathcal{L}$}
\State{$\rho_{\max},\ell^{\ast},a^{\ast}_\ell \gets 0$}
\For{$\ell \in \mathcal{L}$}
\If{$\rho_{\max} < \rho\left(G_\ell\right)$}
\State{}\Comment{Find maximum $\rho$-factor among all chains, the corresponding chain and CAV}
\State{$\rho_{\max} \gets \rho\left(G_\ell\right)$~;~$\ell^{\ast} \gets \ell$~;~$a^{\ast}_\ell \gets a_\ell$} 
\EndIf
\EndFor
\State{$subSequence$ $\gets \emptyset$} 
\While{\texttt{True}}
\If { $\exists~i\in V_{\ell^{\ast}}$ such that $(i,a^{\ast}_\ell) \in E_{\ell^{\ast}}$}
\State{$subSequence$.PushFront($i$) }\Comment{Add $i$ to the front of the subsequence}
\State{Remove $(i,a^{\ast}_\ell)$ from $E_{\ell^{\ast}}$}
\State{Remove $i$ from $V_{\ell^{\ast}}$}
\State{$a^{\ast}_\ell \gets i$}
\Else 
\State{\texttt{Break}}
\EndIf
\EndWhile
\State{$sequence$.PushBack($subSequence$)}\Comment{{Add subsequence to the back of the sequence}}
\EndWhile
\State{\Return $sequence$}
\end{algorithmic} \label{Alg:Re-sequencing Algorithm}
\end{algorithm}

\begin{theorem}
Under Assumption \ref{Alg:Re-sequencing Algorithm}, the decision sequence of $N(t)$ CAVs, which is the optimal solution to $(1~|~G~|~\sum\limits_{i=1}^{N(t)} w_i~C_i^s)$, is computed using Algorithm \ref{Alg:Re-sequencing Algorithm}.
\end{theorem}
\begin{proof}
Among all the disjoint chains of all paths, let $\rho_{\max}$ and $a^{\ast}_\ell$ be the maximum $\rho$-factor and the corresponding CAV determining it, respectively. Namely, $\ell^{\ast}$ is the associated path with $\rho_{\max}$ and $a^{\ast}_\ell$. For every $\ell \in \mathcal{L}\setminus\{\ell^{\ast}\}$, we have
\begin{equation} \label{eq:rho_max}
    \rho_{\max} >  \dfrac{\sum_{j\in V_{\ell}}^{a_{\ell}}~w_j}{\sum_{j\in V_{\ell}}^{a_{\ell}}~P_j},
\end{equation}
where $a_{\ell}$ is the CAV determining the $\rho$-factor of a chain $G_{\ell}=\left(V_{\ell},E_{\ell}\right)$ (lines (3)-(9) in the algorithm).
From \eqref{eq:rho_max} and \cite[Lemma 3.1.2.]{pinedo2016scheduling}, the chain $G_{\ell^\ast}$, should be processed first. From Lemma \ref{lem:rofactor}, all the CAVs in the chain $G_{\ell^\ast}$, should be processed until CAV $a^{\ast}_\ell$ one after another without any interruption by CAVs from other chains ({lines (11)-(20) in the algorithm}). All processed CAVs get removed from their corresponding chain (lines (14) and (15) in the algorithm), and then the process will be repeated until no CAVs remained unprocessed. 
\end{proof}

\section{Simulation Results}\label{sec:Results}
We evaluate the effectiveness of our framework in simulation through several scenarios. In all scenarios, we consider CAVs entering the control zone from six different paths shown in Fig. \ref{fig:intersection}, where the length of control zone for straight and turning paths are $212$ m and $215$ m, respectively. The CAVs enter the control zone with initial speed uniformly distributed between $12$ m/s to $17$ m/s from each entry with equal traffic volumes varying from $800$ to $2400$ veh/h. Videos from our simulation analysis can be found at the supplemental site, \url{https://sites.google.com/view/ud-ids-lab/RPRS}.

For the first scenario, we demonstrate the effects of replanning mechanism to respond to the deviations from the previous planned trajectory. We consider $24$ CAVs entering the control zone with the rate of $2,400$ veh/h per path, and replanning occurs every time the new CAVs enter the control zone. To only consider the effects of replanning, we set the decision sequence of CAVs at each replanning instance to be based on FCFS sequence. 
At each replanning instance, we consider CAVs observe their current position and speed with deviation uniformly distributed in the range of $[-2,2]$ m and $[-0.2,0.2]$ m/s, respectively. The position trajectory of CAVs traveling from westbound to eastbound is visualized in Fig. \ref{fig:TvPError} in the presence of deviation. The CAVs' positions in this path is denoted by a solid line, while their corresponding rear-end constraints are shown with a dashed line. The CAVs from other paths that have the potential for lateral collision with CAVs in this path are shown with a square and vertical bar showing their safety time headway. Figure \ref{fig:TvPError} shows that by using our replanning framework, CAVs respond to the observation made at each replanning point, and they adjust their trajectory to ensure safety.

\begin{figure}
    \centering
\includegraphics[width=0.9\linewidth]{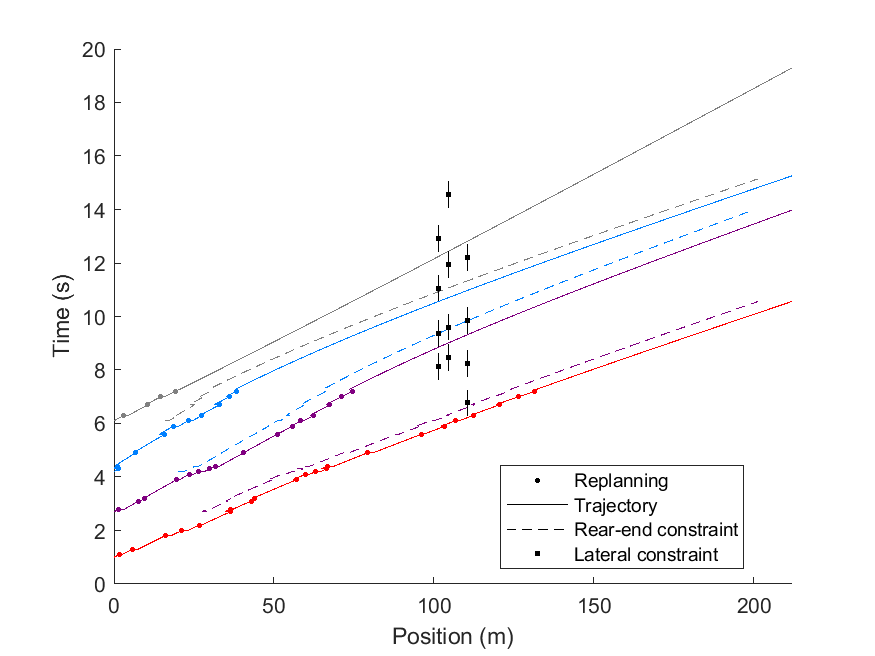}  \caption{Time vs position in the presence of deviation.}
    \label{fig:TvPError}
\end{figure}

For the second scenario, we show the change in average travel time of CAVs within our proposed framework compared to the baseline case for different traffic volumes $2,400$ veh/h and $1,200$ veh/h per path. In the baseline case, CAVs only solve their optimization problem upon entering the control zone based on FCFS, while in our proposed framework, CAVs replan based on the new decision sequence as a new CAV enters the system. In this scenario, we assumed that all CAVs have the same weights, and for each traffic volume, we performed $30$ simulations with different random seeds. The results are presented in Fig. \ref{fig:avgTravelTime}, and it can be seen how resequencing CAVs affects the average travel time. As the traffic flow increases, the change in average travel time varies more, highlighting the importance of decision sequence in influencing the traffic throughput.

\begin{figure}
    \centering
\includegraphics[width=0.92\linewidth]{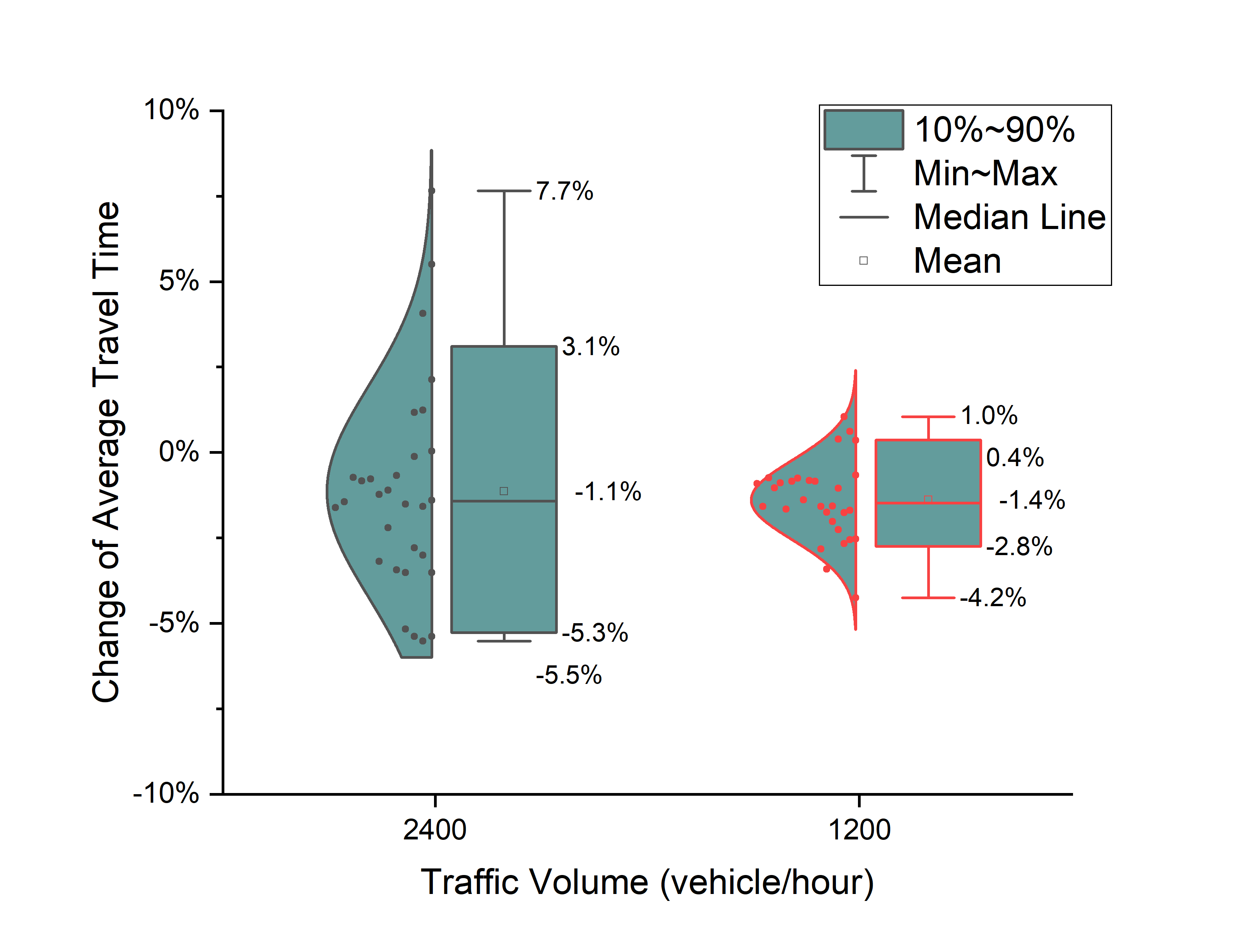}   \caption{Change in average travel time compared to the FCFS decision sequence for different traffic volumes.}
    \label{fig:avgTravelTime}
\end{figure}

For our last scenario, we demonstrate the change in weighted average travel time of CAVs within our proposed framework compared to the baseline case, at different traffic volumes $2,400$ veh/h, $1,200$ veh/h, and $800$ veh/h in Fig. \ref{fig:avgTravelTimediffWeights}.
Similar to the previous scenario, in our proposed framework, CAVs replan based on the new decision sequence as a new CAV enters the system. We performed $30$ simulations with different random seeds for each traffic density. In this scenario, we consider that all CAVs' weights are inversely proportional to the size of their compact set of the feasible solution. 
After performing $30$ different simulations for each traffic flow, our resequencing framework based on the minimum exit time is shown to improve the travel time on average by about $2$\%. 

It should be noted that in 2017, congestion in urban areas across the U.S. led to drivers collectively spending an extra 8.8 billion hours on the road and purchasing an additional 3.3 billion gallons of fuel, ultimately resulting in a \$166 billion expense \cite{Schrank2019}. Thus, $2\%$ improvement on average travel time in the scale of transportation network by only changing the decision sequence can be quite substantial.
Additionally, the main benefit of our approach lies in providing a systematic framework to relax the FCFS sequence in decision making. This would be useful if one needs to prioritize some CAVs to other CAVs, such as giving higher priority to vehicles with higher passenger capacity or emergency vehicles.

By formulating the resequencing problem as a scheduling problem, we find the optimal solution to the scheduling problem. However, this optimal schedule is not the optimal solution, which minimizes the average of actual travel time of all CAVs. The actual travel time of CAVs depends on the decision sequence order, and finding this optimal sequence of decision making is a combinatorial problem, which is an NP-hard problem \cite{de2013autonomous}. However, the algorithm employed in this paper depends on a simple sort which can be done in $O(n\log(n))$ \cite{pinedo2016scheduling}. Thus, if the sole purpose is to improve the average travel time of all CAVs, one can find the decision sequence using our proposed framework and compare it with FCFS policy, and only choose the decision sequence based on the minimum exit time if it improves the performance.

\begin{figure}
    \centering
\includegraphics[width=0.92\linewidth]{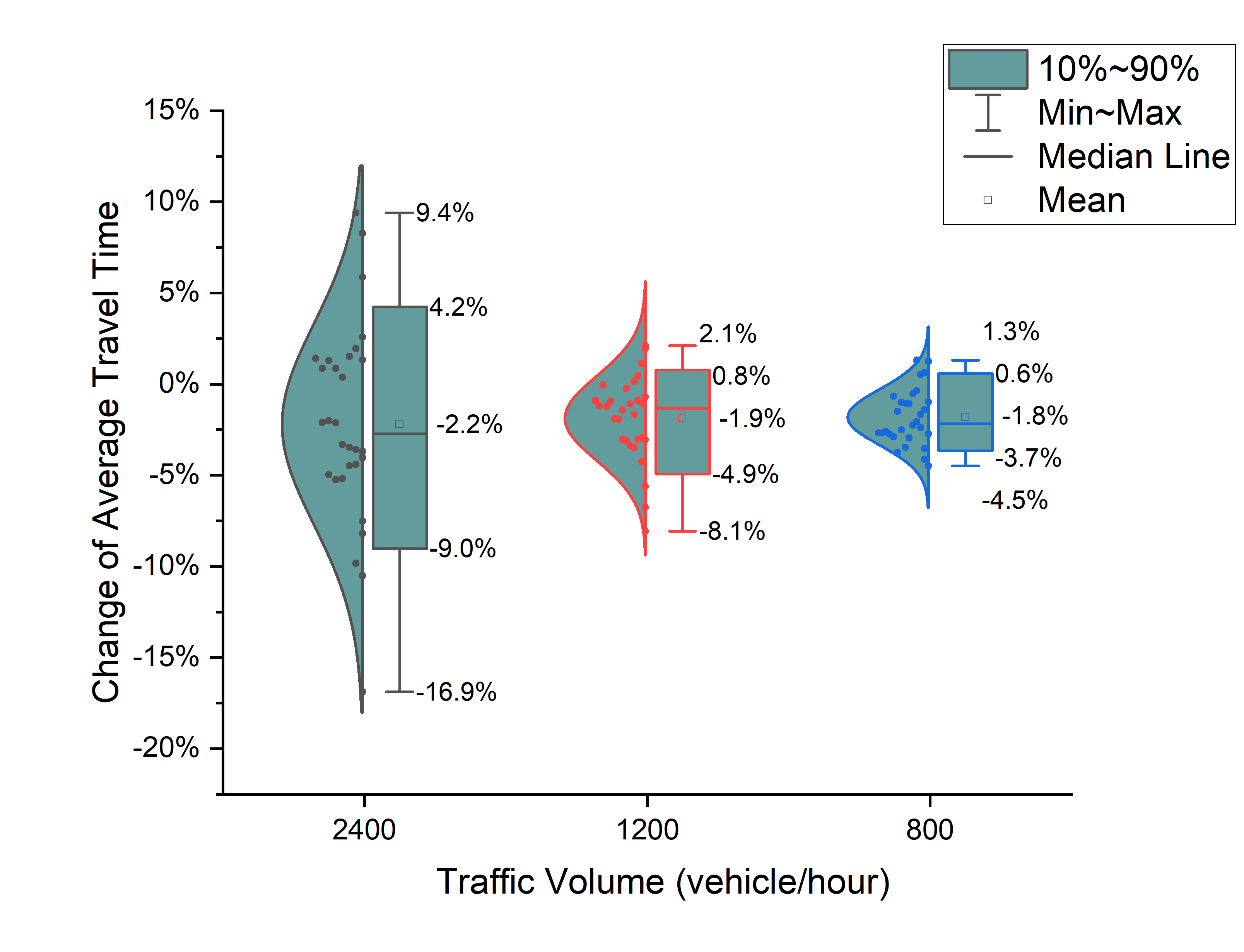}   \caption{Change in weighted average travel time compared to the FCFS decision sequence for different traffic volumes.}
    \label{fig:avgTravelTimediffWeights}
\end{figure}

\section{Concluding Remarks and Discussion} \label{sec:Conclusion}

In this paper, we extended the framework we developed earlier \cite{Malikopoulos2020} for coordination of CAVs at a signal-free intersection to integrate replanning in a time-driven or event-driven manner. This embedded replanning aims at introducing indirect feedback into the coordination framework to respond to the unexpected changes in the system to some extent. Using the theory of the job-shop scheduling problem, we further enhanced our decentralized coordination framework by introducing a priority-aware resequencing mechanism, which designates the order of decision making.
The work that we reported on this paper advances the state of the art in a way that relaxes the FCFS decision making sequence of the CAVs.
Moreover, in our resequencing framework, we can have different weights representing the priorities for CAVs based on the application. In this paper, we chose the weights of the CAVs to be inversely proportional to the size of their compact set of feasible solutions. We finally demonstrated the effectiveness of our proposed approach through several numerical simulations.

\bibliographystyle{IEEEtran.bst} 
\bibliography{reference/IDS_Publications_11122021.bib, reference/ref.bib}

\begin{thebibliography}{10}
\providecommand{\url}[1]{#1}
\csname url@rmstyle\endcsname
\providecommand{\newblock}{\relax}
\providecommand{\bibinfo}[2]{#2}
\providecommand\BIBentrySTDinterwordspacing{\spaceskip=0pt\relax}
\providecommand\BIBentryALTinterwordstretchfactor{4}
\providecommand\BIBentryALTinterwordspacing{\spaceskip=\fontdimen2\font plus
\BIBentryALTinterwordstretchfactor\fontdimen3\font minus
  \fontdimen4\font\relax}
\providecommand\BIBforeignlanguage[2]{{%
\expandafter\ifx\csname l@#1\endcsname\relax
\typeout{** WARNING: IEEEtran.bst: No hyphenation pattern has been}%
\typeout{** loaded for the language `#1'. Using the pattern for}%
\typeout{** the default language instead.}%
\else
\language=\csname l@#1\endcsname
\fi
#2}}

\bibitem{Malikopoulos2017}
A.~A. Malikopoulos, C.~G. Cassandras, and Y.~Zhang, ``A decentralized
  energy-optimal control framework for connected automated vehicles at
  signal-free intersections,'' \emph{Automatica}, vol.~93, pp. 244--256, 2018.

\bibitem{Rios-Torres2017}
J.~Rios-Torres and A.~A. Malikopoulos, ``A survey on the coordination of
  connected and automated vehicles at intersections and merging at highway
  on-ramps,'' \emph{IEEE Transactions on Intelligent Transportation Systems},
  vol.~18, no.~5, pp. 1066--1077, 2016.

\bibitem{xu2019cooperative}
H.~Xu, Y.~Zhang, L.~Li, and W.~Li, ``Cooperative driving at unsignalized
  intersections using tree search,'' \emph{IEEE Transactions on Intelligent
  Transportation Systems}, vol.~21, no.~11, pp. 4563--4571, 2019.

\bibitem{xu2020bi}
H.~Xu, Y.~Zhang, C.~G. Cassandras, L.~Li, and S.~Feng, ``A bi-level cooperative
  driving strategy allowing lane changes,'' \emph{Transportation research part
  C: emerging technologies}, vol. 120, p. 102773, 2020.

\bibitem{guney2020scheduling}
M.~A. Guney and I.~A. Raptis, ``Scheduling-based optimization for motion
  coordination of autonomous vehicles at multilane intersections,''
  \emph{Journal of Robotics}, vol. 2020, 2020.

\bibitem{hult2018optimal}
R.~Hult, M.~Zanon, S.~Gros, and P.~Falcone, ``Optimal coordination of automated
  vehicles at intersections: Theory and experiments,'' \emph{IEEE Transactions
  on Control Systems Technology}, vol.~27, no.~6, pp. 2510--2525, 2018.

\bibitem{chalaki2020TITS}
B.~Chalaki and A.~A. Malikopoulos, ``Time-optimal coordination for connected
  and automated vehicles at adjacent intersections,'' \emph{IEEE Transactions
  on Intelligent Transportation Systems}, pp. 1--16, 2021.

\bibitem{fayazi2018mixed}
S.~A. Fayazi and A.~Vahidi, ``Mixed-integer linear programming for optimal
  scheduling of autonomous vehicle intersection crossing,'' \emph{IEEE
  Transactions on Intelligent Vehicles}, vol.~3, no.~3, pp. 287--299, 2018.

\bibitem{chalaki2020TCST}
B.~Chalaki and A.~A. Malikopoulos, ``Optimal control of connected and automated
  vehicles at multiple adjacent intersections,'' \emph{IEEE Transactions on
  Control Systems Technology}, pp. 1--13, 2021.

\bibitem{Malikopoulos2020}
A.~A. Malikopoulos, L.~E. Beaver, and I.~V. Chremos, ``Optimal time trajectory
  and coordination for connected and automated vehicles,'' \emph{Automatica},
  vol. 125, no. 109469, 2021.

\bibitem{zhang2019decentralized}
Y.~Zhang and C.~G. Cassandras, ``Decentralized optimal control of connected
  automated vehicles at signal-free intersections including comfort-constrained
  turns and safety guarantees,'' \emph{Automatica}, vol. 109, p. 108563, 2019.

\bibitem{kim2014mpc}
K.-D. Kim and P.~R. Kumar, ``An {MPC}-based approach to provable system-wide
  safety and liveness of autonomous ground traffic,'' \emph{IEEE Transactions
  on Automatic Control}, vol.~59, no.~12, pp. 3341--3356, 2014.

\bibitem{campos2014cooperative}
G.~R. Campos, P.~Falcone, H.~Wymeersch, R.~Hult, and J.~Sj{\"o}berg,
  ``Cooperative receding horizon conflict resolution at traffic
  intersections,'' in \emph{53rd IEEE Conference on Decision and
  Control}.\hskip 1em plus 0.5em minus 0.4em\relax IEEE, 2014, pp. 2932--2937.

\bibitem{kloock2019distributed}
M.~Kloock, P.~Scheffe, S.~Marquardt, J.~Maczijewski, B.~Alrifaee, and
  S.~Kowalewski, ``Distributed model predictive intersection control of
  multiple vehicles,'' in \emph{2019 IEEE Intelligent Transportation Systems
  Conference (ITSC)}.\hskip 1em plus 0.5em minus 0.4em\relax IEEE, 2019, pp.
  1735--1740.

\bibitem{de2013autonomous}
G.~R. de~Campos, P.~Falcone, and J.~Sj{\"o}berg, ``Autonomous cooperative
  driving: a velocity-based negotiation approach for intersection crossing,''
  in \emph{16th International IEEE Conference on Intelligent Transportation
  Systems (ITSC 2013)}.\hskip 1em plus 0.5em minus 0.4em\relax IEEE, 2013, pp.
  1456--1461.

\bibitem{alrifaee2016coordinated}
B.~Alrifaee, F.-J. He{\ss}eler, and D.~Abel, ``Coordinated non-cooperative
  distributed model predictive control for decoupled systems using graphs,''
  \emph{IFAC-PapersOnLine}, vol.~49, no.~22, pp. 216--221, 2016.

\bibitem{xiao2020decentralized}
W.~Xiao and C.~G. Cassandras, ``Decentralized optimal merging control for
  connected and automated vehicles with optimal dynamic resequencing,'' in
  \emph{2020 American Control Conference (ACC)}.\hskip 1em plus 0.5em minus
  0.4em\relax IEEE, 2020, pp. 4090--4095.

\bibitem{chalaki2020experimental}
B.~Chalaki, L.~E. Beaver, and A.~A. Malikopoulos, ``Experimental validation of
  a real-time optimal controller for coordination of cavs in a multi-lane
  roundabout,'' in \emph{31st IEEE Intelligent Vehicles Symposium (IV)}, 2020,
  pp. 504--509.

\bibitem{bryson1975applied}
A.~E. Bryson and Y.~C. Ho, \emph{Applied optimal control: optimization,
  estimation and control}.\hskip 1em plus 0.5em minus 0.4em\relax CRC Press,
  1975.

\bibitem{chalaki2021RobustGP}
B.~Chalaki and A.~A. Malikopoulos, ``Robust learning-based trajectory planning
  for emerging mobility systems,'' \emph{arXiv preprint arXiv:2103.03313},
  2021.

\bibitem{pinedo2016scheduling}
M.~L. Pinedo, \emph{Scheduling: theory, algorithms, and systems}.\hskip 1em
  plus 0.5em minus 0.4em\relax Springer, 2016.

\bibitem{Schrank2019}
B.~Schrank, B.~Eisele, and T.~Lomax, ``{2019 Urban Mobility Scorecard},'' Texas
  A\& M Transportation Institute, Tech. Rep., 2019.

\end{thebibliography}

\end{document}